\documentclass[a4paper,10pt]{article}
\usepackage{amssymb}
\usepackage{amsmath}
\usepackage{amsthm}
\usepackage[utf8]{inputenc}
\usepackage{array}
\usepackage[T1]{fontenc}
\usepackage{lmodern}
\usepackage{pgf,tikz}
\usetikzlibrary{arrows}
\usepackage{vmargin}
\setmarginsrb{3.5cm}{2cm}{3.5cm}{2cm}{1cm}{1cm}{1cm}{1cm}

\title{A lower bound on the order of the largest induced forest in
  planar graphs with high girth \footnote{This research was partially supported by ANR EGOS project, under contract ANR-12-JS02-002-01.}}

\usepackage{hyperref}
\usepackage[affil-it]{authblk}

\author[a]{François Dross}
\author[b]{Mickael Montassier}
\author[c]{Alexandre Pinlou}

\affil[a]{{\small ENS de Lyon, LIRMM}} \affil[b]{{\small Université de
    Montpellier, LIRMM}} \affil[c]{{\small Université Paul Valery Montpellier,
    LIRMM\medskip}} \affil[ ]{{\small 161 rue Ada, 34095 Montpellier
    Cedex 5, France}} \affil[
]{\href{mailto:francois.dross@ens-lyon.fr,mickael.montassier@lirmm.fr,alexandre.pinlou@lirmm.fr}{\small{francois.dross@ens-lyon.fr,\{mickael.montassier,alexandre.pinlou\}@lirmm.fr}}}

\begin{document}

\definecolor{ttqqqq}{rgb}{0.06666666666666667,0.06666666666666667,0.06666666666666667}
\definecolor{yqyqyq}{rgb}{0.5019607843137255,0.5019607843137255,0.5019607843137255}
\definecolor{qqqqff}{rgb}{0.3333333333333333,0.3333333333333333,0.3333333333333333}
\definecolor{aqaqaq}{rgb}{0.6274509803921569,0.6274509803921569,0.6274509803921569}

\maketitle
\newtheorem{theo}{Theorem}
\newtheorem{cor}[theo]{Corollary}
\newtheorem{lemm}[theo]{Lemma}
\newtheorem{prop}[theo]{Property}
\newtheorem{obs}[theo]{Observation}
\newtheorem{conj}[theo]{Conjecture}
\newtheorem{claim}[theo]{Claim}

\begin{abstract}
We give here new upper bounds on the size of a smallest feedback vertex set in planar graphs with high girth. In particular, we prove that a planar graph with girth $g$ and size $m$ has a feedback vertex set of size at most $\frac{4m}{3g}$, improving the trivial bound of $\frac{2m}{g}$. We also prove that every $2$-connected graph with maximum degree $3$ and order $n$ has a feedback vertex set of size at most $\frac{n+2}{3}$. 
\end{abstract}

\section{Introduction}

In this article we only consider finite simple graphs. 

Let $G$ be a graph. A \emph{feedback vertex set} or \emph{decycling
  set} $S$ of $G$ is a subset of the vertices of $G$ such that
removing the vertices of $S$ from $G$ yields an acyclic graph. Thus
$S$ is a feedback vertex set of $G$ if and only if the graph induced by
$V(G) \backslash S$ in $G$ is an induced forest of $G$. The {\sc
  feedback vertex set decision problem} (which consists of, given a
graph $G$ and an integer $k$, deciding whether there is a decycling
set of $G$ of size $k$) is known to be NP-complete, even restricted to
the case of planar graphs, bipartite graphs or perfect graphs
\cite{Karp}. It is thus legitimate to seek bounds for the size of a
decycling set or an induced forest. The smallest size of a decycling
set of $G$ is called the \emph{decycling number} of $G$, and the
highest order of an induced forest of $G$ is called the \emph{forest
  number} of $G$, denoted respectively by $\phi(G)$ and $a(G)$. Note
that the sum of the decycling number and the forest number of $G$ is
equal to the order of $G$ (i.e. $|V(G)| = a(G) + \phi(G)$).

Mainly, the community focuses on the following challenging conjecture
due to Albertson and Berman \cite{AlbertsonBerman}:

\begin{conj}[\emph{Albertson and Berman \cite{AlbertsonBerman}}] \label{alb}
Every planar graph $G$ of order $n$ admits an induced forest of order at
least $\frac{n}{2}$, that is $a(G) \ge \frac{n}{2}$.
\end{conj}
Conjecture~\ref{alb}, if true, would be tight (for $n \ge 3$ multiple
of $4$) because of the disjoint union of complete graphs on four
vertices (Akiyama and Watanabe \cite{Akiyama} gave examples showing
that the conjecture differs from the optimal by at most one half for
all $n$), and would imply that every planar graph has an independent
set on at least a quarter of its vertices, the only known proof of
which relies on the Four-Color Theorem.  The best known lower bound to
date for the forest number of a planar graph is due to Borodin and is
a consequence of the acyclic $5$-colorability of planar graphs
\cite{Borodin}. We recall that an acyclic $k$-coloring is a proper
vertex coloring using $k$ colors such that the graph induced by the
vertices of any two color classes is a forest. From Borodin's result one
can obtain the following theorem:

\begin{theo}[\emph{Borodin \cite{Borodin}}]
Every planar graph of order $n$ admits an induced forest of order at
least $\frac{2n}{5}$.
\end{theo}

Hosono \cite{Hosono} showed the following theorem and showed that the
bound is tight.

\begin{theo}[\emph{Hosono \cite{Hosono}}] \label{hos}
Every outerplanar graph of order $n$ admits an induced forest of order
at least $\frac{2n}{3}$.
\end{theo}

The tightness of the bound is shown by the example in
Figure~\ref{bababibelba}.

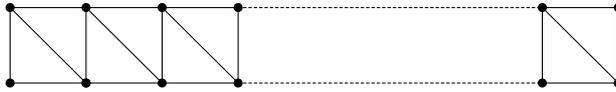
\begin{figure}[t]
\begin{center}
\begin{tikzpicture}[line cap=round,line join=round,>=triangle 45,x=1.0cm,y=1.0cm]
\clip(-0.25,-0.25) rectangle (8.25,1.25);
\draw (0.0,-0.0)-- (0.0,1.0);
\draw (1.0,0.0)-- (1.0,1.0);
\draw (2.0,0.0)-- (2.0,1.0);
\draw (3.0,0.0)-- (3.0,1.0);
\draw (7.0,0.0)-- (7.0,1.0);
\draw (8.0,0.0)-- (8.0,1.0);
\draw (0.0,1.0)-- (1.0,1.0);
\draw (1.0,1.0)-- (2.0,1.0);
\draw (2.0,1.0)-- (3.0,1.0);
\draw (0.0,-0.0)-- (1.0,-0.0);
\draw (1.0,-0.0)-- (2.0,0.0);
\draw (2.0,0.0)-- (3.0,0.0);
\draw (0.0,1.0)-- (1.0,-0.0);
\draw (1.0,1.0)-- (2.0,0.0);
\draw (2.0,1.0)-- (3.0,0.0);
\draw [dash pattern=on 1pt off 1pt] (3.0,1.0)-- (7.0,1.0);
\draw [dash pattern=on 1pt off 1pt] (3.0,0.0)-- (7.0,0.0);
\draw (7.0,1.0)-- (8.0,1.0);
\draw (7.0,0.0)-- (8.0,0.0);
\draw (7.0,1.0)-- (8.0,0.0);
\begin{scriptsize}
\draw [fill=black] (0.0,-0.0) circle (1.5pt);
\draw [fill=black] (0.0,1.0) circle (1.5pt);
\draw [fill=black] (1.0,-0.0) circle (1.5pt);
\draw [fill=black] (1.0,0.0) circle (1.5pt);
\draw [fill=black] (1.0,1.0) circle (1.5pt);
\draw [fill=black] (2.0,0.0) circle (1.5pt);
\draw [fill=black] (2.0,1.0) circle (1.5pt);
\draw [fill=black] (3.0,0.0) circle (1.5pt);
\draw [fill=black] (3.0,1.0) circle (1.5pt);
\draw [fill=black] (7.0,0.0) circle (1.5pt);
\draw [fill=black] (7.0,1.0) circle (1.5pt);
\draw [fill=black] (8.0,0.0) circle (1.5pt);
\draw [fill=black] (8.0,1.0) circle (1.5pt);
\end{scriptsize}
\end{tikzpicture}
\end{center}
\caption{An outerplanar graph $G$ with $a(G) = \frac{2|V(G)|}{3}$. \label{bababibelba}}
\end{figure}

Akiyama and Watanabe \cite{Akiyama}, and Albertson and Haas
\cite{Albertson} independently raised the following conjecture:

\begin{conj}[\emph{Akiyama and Watanabe \cite{Akiyama}, and Albertson and Haas \cite{Albertson}}] \label{aki}
Every bipartite planar graph of order $n$ admits an induced forest of
order at least $\frac{5n}{8}$.
\end{conj}

This conjecture, if true, would be tight for $n$ multiple of $8$: for
example, if $G$ is the disjoint union of $k$ cubes, then we have $a(G)
= 5k$ and $G$ has order $8k$ (see Figure~\ref{bababa}). Motivated by
Conjecture~\ref{aki}, Alon~\cite{Alon2003} proved the following
theorem using probabilistic methods:

\begin{figure}[t]
\begin{center}
\definecolor{cqcqcq}{rgb}{0.65,0.65,0.65}
\begin{tikzpicture}[line cap=round,line join=round,>=triangle 45,x=0.5cm,y=0.5cm]
\draw (0.6879722606657073,6.03171952230884)-- (5.999242097886516,6.031719522308842);
\draw [color=cqcqcq] (5.999242097886516,6.031719522308842)-- (5.999242097886509,0.7204496850880373);
\draw [color=cqcqcq] (5.999242097886509,0.7204496850880373)-- (0.6879722606657072,0.7204496850880311);
\draw (0.6879722606657072,0.7204496850880311)-- (0.6879722606657073,6.03171952230884);
\draw (0.6879722606657073,6.03171952230884)-- (2.243607179276111,4.476084603698435);
\draw [color=cqcqcq] (2.243607179276111,4.476084603698435)-- (4.4436071792761105,4.476084603698435);
\draw [color=cqcqcq] (4.4436071792761105,4.476084603698435)-- (4.443607179276109,2.276084603698438);
\draw [color=cqcqcq] (4.443607179276109,2.276084603698438)-- (2.243607179276111,2.2760846036984357);
\draw [color=cqcqcq] (2.243607179276111,2.2760846036984357)-- (2.243607179276111,4.476084603698435);
\draw [color=cqcqcq] (2.243607179276111,2.2760846036984357)-- (0.6879722606657072,0.7204496850880311);
\draw [color=cqcqcq] (4.443607179276109,2.276084603698438)-- (5.999242097886509,0.7204496850880373);
\draw [color=cqcqcq] (4.4436071792761105,4.476084603698435)-- (5.999242097886516,6.031719522308842);
\begin{scriptsize}
\draw [fill=black] (2.243607179276111,4.476084603698435) circle (1.5pt);
\draw [color=cqcqcq][fill=cqcqcq] (2.243607179276111,2.2760846036984357) circle (1.5pt);
\draw [color=cqcqcq][fill=cqcqcq] (4.4436071792761105,4.476084603698435) circle (1.5pt);
\draw [fill=black] (4.443607179276109,2.276084603698438) circle (1.5pt);
\draw [fill=black] (0.6879722606657073,6.03171952230884) circle (1.5pt);
\draw [fill=black] (5.999242097886516,6.031719522308842) circle (1.5pt);
\draw [color=cqcqcq][fill=cqcqcq] (5.999242097886509,0.7204496850880373) circle (1.5pt);
\draw [fill=black] (0.6879722606657072,0.7204496850880311) circle (1.5pt);
\end{scriptsize}
\end{tikzpicture}
\end{center}
\caption{The cube $Q$ admits an induced forest on five of its vertices,
  but no induced forest on six or more of its
  vertices, i.e. $a(Q) = 5$. \label{bababa}}
\end{figure}
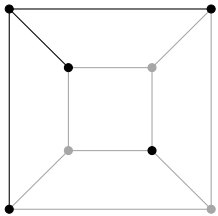

\begin{theo}[\emph{Alon~\cite{Alon2003}}]
There exist some absolute constants $b > 0$ and $b' > 0$ such that:

\begin{itemize}
\item For every bipartite graph $G$ with $n$ vertices and average
  degree at most $d$ ($\ge 1$), $a(G) \ge (\frac{1}{2} + e^{-bd^2})n$.

\item For every $d \ge 1$ and all sufficiently large $n$, there exists
  a bipartite graph with $n$ vertices and average degree at most $d$
  such that $a(G) \le (\frac{1}{2} + e^{-b'\sqrt{d}})n$.
\end{itemize}
\end{theo}

The lower bound was later improved by Conlon et
al. \cite{conlonessays} to $a(G) \ge (1/2 + e^{-b''d})n$ for a
constant $b''$.

Conjecture~\ref{aki} also led to researches for lower bounds of the
forest number of triangle-free planar graphs (as a superclass of
bipartite planar graphs). Alon {\em et al.}~\cite{Alon} proved the following
theorem and corollary:

\begin{theo}[\emph{Alon et al.~\cite{Alon}}]\label{bAlon}
Every triangle-free graph of order $n$ and size $m$ admits an induced
forest of order at least $n - \frac{m}{4}$.
\end{theo}

\begin{cor}[\emph{Alon et al.~\cite{Alon}}]
Every triangle-free cubic graph of order $n$ admits an induced forest
of order at least $\frac{5n}{8}$.
\end{cor}

Theorem~\ref{bAlon} is tight because of the union of cycles of length
$4$.

\bigskip The \emph{girth} of a graph is the length of a shortest cycle.
A forest has infinite girth.
In a planar graph with girth at least $g$, order $n$, and size $m$ with
at least one cycle, the number of faces is at most $2m/g$ (since all the
faces' boundaries have length at least $g$). Then, by Euler's formula,
$2m/g \ge m - n + 2$, and thus $m \le (g/(g-2)) (n - 2)$. In
particular, triangle-free planar graphs of order $n \ge 3$ have size
at most $2n-4$. As a consequence of Theorem~\ref{bAlon}, for a
triangle-free planar graph $G$ of order $n$, $a(G) \ge n/2$. Salavatipour
proved a better lower bound \cite{Salavatipour}: $a(G)\ge
\frac{17n+24}{32}$. In a companion paper, the authors strengthen this
bound as follows:

\begin{theo}[{\em \cite{girth4-5}}] \label{comain}
Every triangle-free planar graph of order $n \ge 1$ admits an induced
forest of order at least $\frac{6n + 7}{11}$.
\end{theo}

Kowalik {\em et al}.~\cite{Kowalik} made the following conjecture on planar
graphs of girth at least $5$:

\begin{conj}[\emph{Kowalik et al.~\cite{Kowalik}}] \label{conj:kowalik}
Every planar graph with girth at least $5$ and order $n$ admits an
induced forest of order at least $7n/10$.
\end{conj}

This conjecture, if true, would be tight for $n$ multiple of $20$, as
shown by the example of the union of dodecahedrons, given by Kowalik
{\em et al}.~\cite{Kowalik} (see Figure~\ref{bibibi}).

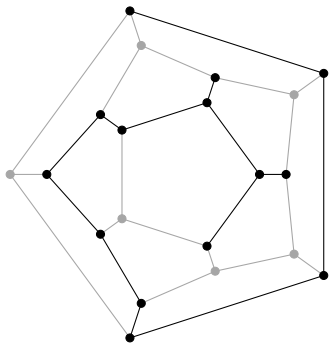
\begin{figure}[htbp]
\begin{center}
\definecolor{cqcqcq}{rgb}{0.65,0.65,0.65}
\begin{tikzpicture}[line cap=round,line join=round,>=triangle 45,x=1.0cm,y=1.0cm]
\clip(-2.6,-2.5) rectangle (2.3,2.5);
\draw (1.0,-0.0)-- (0.3090169943749475,0.9510565162951535);
\draw (0.3090169943749475,0.9510565162951535)-- (-0.8090169943749473,0.5877852522924731);
\draw [color=cqcqcq] (-0.8090169943749473,0.5877852522924731)-- (-0.8090169943749475,-0.5877852522924731);
\draw (0.30901699437494745,-0.9510565162951536)-- (1.0,-0.0);
\draw [color=cqcqcq] (0.30901699437494745,-0.9510565162951536)-- (-0.8090169943749475,-0.5877852522924731);
\draw [color=cqcqcq] (1.4543870332530722,1.056674031837904)-- (1.8443253811383944,1.3399808248767398);
\draw (1.8443253811383944,1.3399808248767398)-- (-0.7044696092807624,2.1681345189236856);
\draw [color=cqcqcq] (-0.7044696092807624,2.1681345189236856)-- (-0.55552641390555,1.709734498543117);
\draw [color=cqcqcq] (-0.7044696092807624,2.1681345189236856)-- (-2.279711543715264,0.0);
\draw [color=cqcqcq] (-2.279711543715264,0.0)-- (-0.7044696092807627,-2.1681345189236856);
\draw (-0.7044696092807627,-2.1681345189236856)-- (1.8443253811383948,-1.3399808248767402);
\draw (1.8443253811383948,-1.3399808248767402)-- (1.8443253811383944,1.3399808248767398);
\draw [color=cqcqcq] (1.4543870332530722,1.056674031837904)-- (0.41694269437650827,1.2832176664280721);
\draw [color=cqcqcq] (0.41694269437650827,1.2832176664280721)-- (-0.55552641390555,1.709734498543117);
\draw [color=cqcqcq] (-0.55552641390555,1.709734498543117)-- (-1.091570145238658,0.7930721328168736);
\draw (-1.091570145238658,0.7930721328168736)-- (-1.7977212386950445,0.0);
\draw (-1.7977212386950445,0.0)-- (-1.091570145238658,-0.7930721328168735);
\draw (-1.091570145238658,-0.7930721328168735)-- (-0.5555264139055502,-1.709734498543117);
\draw [color=cqcqcq] (-0.5555264139055502,-1.709734498543117)-- (0.4169426943765081,-1.2832176664280721);
\draw [color=cqcqcq] (0.4169426943765081,-1.2832176664280721)-- (1.4543870332530724,-1.0566740318379042);
\draw [color=cqcqcq] (1.4543870332530724,-1.0566740318379042)-- (1.3492549017242996,-0.0);
\draw [color=cqcqcq] (1.3492549017242996,-0.0)-- (1.4543870332530722,1.056674031837904);
\draw (1.3492549017242996,-0.0)-- (1.0,-0.0);
\draw (0.3090169943749475,0.9510565162951535)-- (0.41694269437650827,1.2832176664280721);
\draw (-1.091570145238658,0.7930721328168736)-- (-0.8090169943749473,0.5877852522924731);
\draw [color=cqcqcq] (-1.7977212386950445,0.0)-- (-2.279711543715264,0.0);
\draw [color=cqcqcq] (-1.091570145238658,-0.7930721328168735)-- (-0.8090169943749475,-0.5877852522924731);
\draw (-0.5555264139055502,-1.709734498543117)-- (-0.7044696092807627,-2.1681345189236856);
\draw [color=cqcqcq] (0.4169426943765081,-1.2832176664280721)-- (0.30901699437494745,-0.9510565162951536);
\draw [color=cqcqcq] (1.4543870332530724,-1.0566740318379042)-- (1.8443253811383948,-1.3399808248767402);
\begin{scriptsize}
\draw [fill=black] (0.3090169943749475,0.9510565162951535) circle (1.5pt);
\draw [fill=black] (1.0,-0.0) circle (1.5pt);
\draw [fill=black] (-0.8090169943749473,0.5877852522924731) circle (1.5pt);
\draw [color=cqcqcq][fill=cqcqcq] (-0.8090169943749475,-0.5877852522924731) circle (1.5pt);
\draw [fill=black] (0.30901699437494745,-0.9510565162951536) circle (1.5pt);
\draw [fill=black] (-1.091570145238658,0.7930721328168736) circle (1.5pt);
\draw [fill=black] (0.41694269437650827,1.2832176664280721) circle (1.5pt);
\draw [fill=black] (1.3492549017242996,-0.0) circle (1.5pt);
\draw [color=cqcqcq][fill=cqcqcq] (0.4169426943765081,-1.2832176664280721) circle (1.5pt);
\draw [fill=black] (-1.091570145238658,-0.7930721328168735) circle (1.5pt);
\draw [color=cqcqcq][fill=cqcqcq] (1.4543870332530722,1.056674031837904) circle (1.5pt);
\draw [color=cqcqcq][fill=cqcqcq] (-0.55552641390555,1.709734498543117) circle (1.5pt);
\draw [color=cqcqcq][fill=cqcqcq] (1.4543870332530724,-1.0566740318379042) circle (1.5pt);
\draw [fill=black] (-0.5555264139055502,-1.709734498543117) circle (1.5pt);
\draw [fill=black] (-1.7977212386950445,0.0) circle (1.5pt);
\draw [fill=black] (1.8443253811383944,1.3399808248767398) circle (1.5pt);
\draw [fill=black] (-0.7044696092807624,2.1681345189236856) circle (1.5pt);
\draw [color=cqcqcq][fill=cqcqcq] (-2.279711543715264,0.0) circle (1.5pt);
\draw [fill=black] (-0.7044696092807627,-2.1681345189236856) circle (1.5pt);
\draw [fill=black] (1.8443253811383948,-1.3399808248767402) circle (1.5pt);
\end{scriptsize}
\end{tikzpicture}
\end{center}
\caption{The dodecahedron $D$ admits an induced forest on fourteen of its
  vertices, but no induced forest on fifteen or more of its
  vertices, i.e. $a(D) = 14$. \label{bibibi}}
\end{figure}


A first step toward Conjecture~\ref{conj:kowalik} was done in a companion paper
\cite{girth4-5}; moreover a generalization for higher girth was given:

\begin{theo}[{\em \cite{girth4-5}}] \label{bcomain}
Every planar graph with girth at least $5$ and order $n\ge 1$ admits
an induced forest of order at least $\frac{44n+50}{69}$.
\end{theo}


\begin{theo}[{\em \cite{girth4-5}}] \label{bcomainbis}
Every planar graph with girth at least $g \ge 5$ and order $n\ge 1$
admits an induced forest of order at least $n -
\frac{(5n-10)g}{23(g-2)}$.
\end{theo}

For planar graphs with given girth, we conjecture the following:

\begin{conj} \label{m/g}
Let $G$ be a planar graph of size $m$ and girth $g$. There exists a
feedback vertex set $S$ of $G$ of size at most $\frac{m}{g}$.
\end{conj}

If Conjecture \ref{m/g} is true, then it is tight for $m$ multiple of
$g$ due to the union of disjoint cycles of length $g$. It is easy to
prove that $G$ admits a feedback vertex set of size at most
$\frac{2m}{g}$ (removing a vertex that is in the boundary of at least
 two faces decreases the number of faces by one, and this can be 
 applied recursively).

\bigskip
The main result of this paper is a first non-trivial step toward
Conjecture~\ref{m/g}:

\begin{theo} \label{main}
Let $G$ be a planar graph of size $m$ and girth $g$. There exists a
feedback vertex set $S$ of $G$ of size at most $\frac{4m}{3g}$.
\end{theo}

Theorem~\ref{main} is the best result so far for $g \ge 7$, and gives 
$a(G) \ge \frac{(3g-10)n+8}{3(g-2)}$ using $m \le (n-2)\frac{g}{g-2}$
(Theorem~\ref{bcomainbis} is better for $g = 6$). We summarize the
previous results in Table \ref{table:01}.

\begin{center}
\begin{table}\label{table:01}
\begin{center}
\begin{tabular}{|c|c|}
\hline
Planar graph with girth $g$ & Forest number \\ \hline
4 & $\frac{6n + 7}{11}$ \\ \hline
5 & $\frac{44n+50}{69}$ \\ \hline
6 & $\frac{31n+30}{46}$ \\ \hline
$g\ge 7$ & $\frac{(3g-10)n+8}{3(g-2)}$ \\ \hline
\end{tabular}
\end{center}
\caption{Lower bounds on the forest numbers for planar graphs with given girth.}
\end{table}
\end{center}

Theorem~\ref{main} will be proven in Section \ref{proofmain}.  For
this, we will use Theorem \ref{subth} (proven in Section
\ref{proofsubth}) that is of independent interest. Let ${\cal
  C}_{2,3^-}$ be the familly of $2$-connected graphs of maximum degree
at most $3$.

\begin{theo} \label{subth}
Every graph in ${\cal C}_{2,3^-}$ of order $n$ has a feedback vertex
set of size at most $\frac{n+2}{3}$.
\end{theo}

Theorem~\ref{subth} is tight for the complete graph on $4$
vertices. Moreover, consider any $3$-regular graph $G$, and consider
the graph $H$ obtained from $G$ by replacing each vertex by a triangle
(as the cube connected cycles obtained from the hypercube). Graph $H$
has $3|V(G)|$ vertices, and cannot have a feedback vertex set of less
than $|V(G)|$ vertices, since a feedback vertex set of $H$ contains at
least one vertex of each added triangle. Hence there is a graph of
order $n$ without a feedback vertex set of size less than
$\frac{n}{3}$ for an arbitrary large $n$.

Finally, if we replace the $2$-connected condition by simply
connected, then $\frac{3n}{8} + \frac{1}{4}$ becomes a tight bound
\cite{Alon}. One can observe that without connected condition, the
disjoint union of complete graphs on four vertices has a smallest
feedback vertex set of size $\frac{n}{2}$.

\bigskip
\noindent{\bf Notations.} Consider $G = (V,E)$. For a set $S \subseteq
V$, let $G - S$ be the graph obtained from $G$ by removing the
vertices of $S$ and all the edges that are incident to a vertex of
$S$. If $x \in V$, then we denote $G - \{x\}$ by $G - x$. For a set
$S$ of vertices such that $S \cap V = \emptyset$, let $G + S$ be the
graph constructed from $G$ by adding the vertices of $S$. If $x \notin
V$, then we denote $G + \{x\}$ by $G + x$. For a set $F$ of pairs of
vertices of $G$ such that $F \cap E = \emptyset$, let $G + F$ be the
graph constructed from $G$ by adding the edges of $F$. If $e$ is a
pair of vertices of $G$ and $e \notin E$, then we denote $G + \{e\}$
by $G + e$. For a set $W \subseteq V$, we denote by $G[W]$ the subgraph
of $G$ induced by $W$.  We call a vertex of degree $d$, at least $d$,
and at most $d$, a \emph{$d$-vertex}, a \emph{$d^+$-vertex}, and a
\emph{$d^-$-vertex} respectively. Similarly, we call a cycle of length
$\ell$, at least $\ell$, and at most $\ell$ a \emph{$\ell$-cycle}, a
\emph{$\ell^+$-cycle}, and a \emph{$\ell^-$-cycle} respectively, and by
extension a face of length $\ell$, at least $\ell$, and at most $\ell$ a
\emph{$\ell$-face}, a \emph{$\ell^+$-face}, and a \emph{$\ell^-$-face}
respectively.  For a face $f$ of a plane graph $G$, we denote the
boundary of $f$ by $G[f]$.  We say that two faces are \emph{adjacent}
if their boundaries share (at least) an edge. We say that two cycles 
are \emph{adjacent} if they share at least an edge. An \emph{edge cut-set} of 
a graph $G$ is a minimal set of edges $F$ such that $G \backslash F$
is disconnected. If an edge cut-set is a singleton, then its element is a 
\emph{cut edge}. A \emph{vertex cut-set} of a graph $G$ is a set $X$ of 
vertices of $G$ such that $G \backslash X$ is disconnected. If a vertex cut-set
 is a singleton, then its element is a \emph{cut vertex}.


\section{Proof of Theorem~\ref{subth} \label{proofsubth}}

We recall that $G=(V,E)$ is called \emph{$k$-connected} if $|V|>k$ and $G -
X$ is connected for every set $X\subseteq V$ with $|X|<k$. Also
$G=(V,E)$ is called \emph{$k$-edge connected} if $|V|>1$ and the deletion of
any set of at most $(k-1)$ edges leads to a connected graph. 

Let us consider $H = (V,E)$ a counter-example to Theorem~\ref{subth}
of minimum order, and $n = |V| \ge 3$ be the order of $H$. Let us prove
some lemmas on the structure of $H$.

\begin{lemm} \label{Hcubic}
Graph $H$ is cubic.
\end{lemm}

\begin{proof} 
Suppose there is a vertex $v$ of degree at most $2$ in $H$. As
$H$ is $2$-connected, $v$ has degree $2$. Let $u$ and $w$ be the
two neighbors of $v$ in $H$. Suppose $uw \notin E$. Let $H'=H -
v + uw$. Since $u$ and $w$ have degree at least $2$ ($H$ is
$2$-connected), $|V(H')| \ge 3$. Then graph $H'$ is in ${\cal
C}_{2,3^-}$ since $H$ is. By minimality of $H$, $H'$ has a
feedback vertex set $S$ of size $|S| \le \frac{n - 1 +2}{3} \le
\frac{n+2}{3}$, and $S$ is also a feedback vertex set of $H$, a
contradiction. Therefore $uw \in E$. If both $u$ and $w$ have
degree $2$, then $H = C_3$ and $H$ admits a
feedback vertex set of size $1\le \frac{n+2}{3} = \frac{5}{3}$,
a contradiction. If one of $u$ and $w$
has degree $2$ and the other one has degree $3$, then $H$ is not
$2$-connected, a contradiction. Therefore both $u$ and $w$ have
degree $3$. Let $u'$ and $w'$ be the third neighbors of $u$ and
$w$ respectively. If $u' = w'$, then $V = \{u,v,w,u'\}$ ($H$ is
$2$-connected), and $H$ admits a feedback vertex set of size
$1\le \frac{n+2}{3} = 2$ ($\{u\}$ for example), a contradiction.
Thus $u'$ and $w'$ are distinct. Suppose $u'w' \in E$. Let $H' =
H - \{u,v,w\}$. If $|V(H')|<3$, then $u'$ and $w'$ are adjacent
vertices of degree $2$ in $H$ and we fall into a previous
case. Therefore $|V(H')| \ge 3$. Then graph $H'$ is in ${\cal
C}_{2,3^-}$ since $H$ is. By minimality of $H$, $H'$ has a
feedback vertex set $S'$ of size $|S'| \le \frac{n - 3 + 2}{3}$.
The set $S = S' \cup \{u\}$ is a feedback vertex set of $H$ of
size $|S| \le \frac{n - 3 + 2}{3} + 1 = \frac{n + 2}{3}$, a
contradiction. Therefore $u'w' \notin E$. Let $H' = H -
\{u,v,w\} + u'w'$. Graph $H'$ is in ${\cal C}_{2,3^-}$ since $H$ is. By
minimality of $H$, $H'$ has a feedback vertex set
$S'$ of size $|S'| \le \frac{n - 3 + 2}{3}$. The set $S = S'
\cup \{u\}$ is a feedback vertex set of $H$ of size $|S| \le
\frac{n - 3 + 2}{3} + 1 = \frac{n + 2}{3}$, a contradiction.
\end{proof}

In the following, we will use the fact that $H$ is cubic without 
referring to Lemma~\ref{Hcubic}.

\begin{lemm} \label{notwotri}
There are no adjacent triangles in $H$.
\end{lemm}

\begin{proof}
Assume that there are two triangles $xyz$ and $xyz'$ sharing an
edge $xy$ in $H$. If $zz' \in E$, then $H = K_4$ ($H$ is
connected), which contradicts the fact that $H$ is a
counter-example to Theorem~\ref{subth}. Therefore $zz' \notin
E$. Let $v$ be the neighbor of $z$ distinct from $x$ and $y$.
Observe that $vz' \notin E$, since $H$ is cubic and $2$-connected.
Let $H' = H - \{x,y,z\} + vz'$. Graph $H'$ is in ${\cal
C}_{2,3^-}$ since $H$ is. By minimality of $H$, $H'$ has a
feedback vertex set $S'$ of size $|S'| \le \frac{n - 3 + 2}{3}$.
The set $S = S' \cup \{x\}$ is a feedback vertex set of $H$ of
size $|S| \le \frac{n - 3 + 2}{3} + 1 = \frac{n + 2}{3}$, a
contradiction.
\end{proof}

\begin{lemm}\label{notrisqua}
There is no triangle that shares an edge with a $4$-cycle in $H$.
\end{lemm}

\begin{proof}
By Lemma~\ref{notwotri}, there is no triangle that shares two edges with a $4$-cycle in $H$. Assume that there are a triangle $xyw$ and a $4$-cycle $vzxy$ that share the edge $xy$. 

Suppose first that there is a vertex $z'$ adjacent to $v$ and $w$. If $zz' \in E$, then $V = \{v,w,x,y,z,z'\}$ ($H$ is connected), i.e. $H$ is the prism, and $\{y,z\}$ is a feedback vertex set of $H$, thus $H$ is not a counter-example to Theorem~\ref{subth}, a contradiction. Therefore $zz' \notin E$. Let $z''$ be the third neighbor of $z'$. Let $H' = H - \{w,y,z'\} + xv + vz''$. Graph $H'$ is in ${\cal C}_{2,3^-}$ since $H$ is. By minimality of $H$, $H'$ admits a feedback vertex set $S'$ of size at most $|S'| \le \frac{n - 3 + 2}{3}$. The set $S = S' \cup \{w\}$ is a feedback vertex set of $H$ of size $|S| \le \frac{n - 3 + 2}{3} + 1 = \frac{n + 2}{3}$, a contradiction.

Therefore there is no vertex adjacent to $v$ and $w$. Let $w'$ be the neighbor of $w$ distinct from $x$ and $y$. Let $H'' = H - \{x,y,w\} + vw'$. Graph $H''$ is in ${\cal C}_{2,3^-}$. By minimality of $H$, $H''$ admits a feedback vertex set $S''$ of size at most $|S''| \le \frac{n - 3 + 2}{3}$. The set $S = S'' \cup \{x\}$ is a feedback vertex set of $H$ of size $|S| \le \frac{n - 3 + 2}{3} + 1 = \frac{n + 2}{3}$, a contradiction.
\end{proof}

\begin{lemm} \label{notwosqua}
There are no two $4$-cycles that share two edges in $H$.
\end{lemm}

\begin{proof}
Let $uvwx$ and $vwxy$ be two $4$-cycles of $H$. Let $u'$, $w'$ and $y'$ be the third neighbors of $u$, $w$ and $y$ respectively. By Lemma~\ref{notwotri}, they are distinct from the vertices defined previously. If $u' = w' = y'$, then $H = K_{3,3}$ admits a feedback vertex set of size $2 \le \frac{6 + 2}{3} = \frac{n + 2}{3}$ (for example $\{u,y\}$), a contradiction. 

Suppose $u' \ne w' \ne y' \ne u'$. Let $H' = H - \{u,v,w,y\} + \{u'x, w'x, y'x\}$. If $H'$ is not $2$-connected, then w.l.o.g. $x$ separates $u'$ and $w'$ in $H'$, and thus $u$ separates $u'$ and $w'$ in $H$, a contradiction. Therefore $H'$ is in ${\cal C}_{2,3^-}$. By minimality of $H$, $H'$ admits a feedback vertex set $S'$ of size at most $\frac{n - 4 + 2}{3}$. The set $S = S' \cup \{v\}$ is a feedback vertex set of $H$ of size $|S'| + 1 \le \frac{n - 4 + 2}{3} + 1 \le \frac{n + 2}{3}$, a contradiction.

Thus w.l.o.g., $u' = y' \ne w'$. Let $z$ be the neighbor of $u'$ distinct from $u$ and $y$. Observe that $z$ is distinct from $w'$ since $H$ is cubic and $2$-connected. Let $H' = H - \{u,v,w,x,y,u'\}$ if $zw' \in E$ and $H' = H - \{u,v,w,x,y,u'\} + zw'$ otherwise. Graph $H'$ is in ${\cal C}_{2,3^-}$ since $H$ is. By minimality of $H$, $H'$ admits a feedback vertex set $S'$ of size at most $\frac{n - 6 + 2}{3}$. The set $S = S' \cup \{v, x\}$ is a feedback vertex set of $H$ of size $|S'| + 2 \le \frac{n - 6 + 2}{3} + 2 \le \frac{n + 2}{3}$, a contradiction.
\end{proof}

\begin{lemm} \label{evco}
For every $k \in \{1,2,3\}$, a graph with maximum degree at most $3$ is
$k$-connected if and only if it is $k$-edge-connected.
\end{lemm}

\begin{proof}
Let $G$ be a graph with maximum degree at most $3$. One can easily check that the result holds for the complete graph on at most four vertices.

Suppose now that $G$ is not complete. Let $C_v$ be a vertex cut-set of $G$ and $C_e$ be a edge cut-set of $G$, both of minimum size. If we show that $|C_v| = |C_e|$, then the lemma holds.

Let $V_1$ and $V_2$ be the vertex sets of the two connected components of
$G - C_e$. We have $V_1 \cup V_2 = V(G)$. By minimality of $|C_e|$, every edge
of $C_e$ has an endvertex in $V_1$ and the other one in $V_2$. Suppose
every vertex of $V_1$ is adjacent to every vertex of $V_2$ in $G$. We
have $|C_e| = |V_1||V_2| \ge |V_1| + |V_2| - 1 =|V(G)| - 1$. Moreover,
for any vertex in $G$, the set of the edges incident to this vertex is
an edge cut-set of $G$. Therefore, since $G$ is not complete, by minimality of
$C_e$, $|C_e| \le |V(G)| - 2$, a contradiction. Therefore there are two
vertices $v_1 \in V_1$ and $v_2 \in V_2$ such that $v_1v_2 \notin
E(G)$. Let $C_v' = \{x \ne v_1|\exists y \in V_2, xy \in C_e\} \cup \{y| v_1y \in
C_e\}$. Note that $|C_v'| = |\{x \ne v_1|\exists y \in V_2, xy \in C_e\}| + 
|\{y| v_1y \in C_e\}| \le |C_e|$. For each edge in $C_e$, one of the
endvertices of this edge is in $C_v'$. As neither $v_1$ nor $v_2$ is in
$C_v'$, $C_v'$ separates $v_1$ from $v_2$ in $G$. Therefore $|C_v| \le |C_v'|$, 
and thus $|C_v| \le |C_e|$.

Let $W_1$ and $W_2$ be the vertex sets of two connected components of $G-C_v$. Let $x \in
C_v$. Since $x$ has degree at most $3$, $x$ has at most one neighbor in $W_1$ or at most
one neighbor in $W_2$, and it has at least one neighbor in $W_1$ and one in $W_2$ by
minimality of $C_v$. Let $y_x$ be the neighbor of $x$ that is in $W_1$ if there is only
one neighbor of $x$ in $W_1$, and the neighbor of $x$ in $W_2$ otherwise, and $e_x =
xy_x$. Observe that this defines a unique edge $e_x$ for every $x \in C_v$. Let $C_e' =
\{e_x|x \in C_v\}$. Assume $C_e'$ does not separate $W_1$ and
$W_2$. There are $v_1 \in W_1$ and $v_2 \in W_2$ such that there is a path $P$ from 
$v_1$ to $v_2$ in $H - C_e'$. Let us consider $v_1$ and $v_2$ such that $P$ has
minimal length. Then there are $w_1$ and $w_2$ in $C_v$ such that $v_1w_1 \in E(P)$
and $v_2w_2 \in E(P)$. If $w_1 = w_2$, then either $v_1w_1 \in C_e'$ or 
$v_2w_2 \in C_e'$, a contradiction. If $w_1 \ne w_2$, then $w_1$ has a neighbor in 
$V(G) \backslash (W_1 \cup W_2)$, so it has only one neighbor in $W_1$, that is 
$v_1$, so $v_1w_1 \in C_e'$, a contradiction. Therefore $C_e'$ separates $W_1$ and 
$W_2$. We have $|C_e'| = |C_v|$, thus $|C_e| \le |C_v|$. Finally, since 
$|C_v| \le |C_e|$, $|C_v| = |C_e|$.
\end{proof}

\begin{lemm} \label{3co}
Graph $H$ is $3$-connected.
\end{lemm}

\begin{proof}
Suppose by contradiction that $H$ is not $3$-connected. By Lemma~\ref{Hcubic}, $|V(H)| \ge 4$. By hypothesis and Lemma~\ref{evco}, $H$ is $2$-edge-connected but not $3$-edge-connected. Let $\{e,f\}$ be an edge cut-set of $H$ that induces two connected components $V_1$ and $V_2$ such that $|V_1|$ is minimum. 

We will now prove the two following properties:

\begin{itemize}
\item{$P_e$:} The deletion of any edge in $H[V_1]$ preserves the $2$-edge connectivity of $H$.

By contradiction, suppose there is an edge $e'$ that has both of its endvertices in $V_1$ such that $H-e'$ is not $2$-edge-connected (but connected since $H$ is $2$-edge-connected). Let $f'$ be a cut edge of $H-e'$. If $f'$ has at least one of its endvertices in $V_1$, then one of the connected components of $H-\{e',f'\}$ is strictly included in $V_1$, a contradiction with the minimality of $|V_1|$. Therefore, $f'$ has both of its endvertices in $V_2$. Neither $e$ nor $f$ is a cut edge of $H-e'$, otherwise we fall into the previous case. Thus $e'$ is not a cut edge of $H[V_1]$. In particular, there is a path in $H \backslash \{f',e'\}$ that connects the two endvertices of $e'$. However, $e'$ is a cut edge of $H \backslash f'$, a contradiction.

\item{$P_v$:} For every vertex $v$ in $V_1$ that has all of its neighbors in $V_1$, $H - v$ is $2$-edge-connected, and thus $2$-connected by Lemma~\ref{evco}.

Suppose there is a vertex $v \in V_1$ that is not incident to an edge of $\{e,f\}$ such that $H - v$ is not $2$-edge-connected. Let $f'$ be a cut edge of $H-v$. As vertex $v$ has degree $3$, there is an edge $e'$ incident to $v$ such that $H-\{e',f'\}$ is disconnected. As $v$ is not incident to an edge of $\{e,f\}$, $e'$ has both of its endvertices in $V_1$, a contradiction with $P_e$.
\end{itemize}

Let $v \in V_1$ and $u \in V_2$ such that $e = uv$. Let $w$ and $x$ be the two neighbors of $v$ distinct from $u$. Vertices $w$ and $x$ are in $V_1$, otherwise w.l.o.g. $f = vw$, and $vx$ is a cut edge of $H$, a contradiction.

Let us show that $wx \notin E$. By contradiction assume that $wx \in E$. Let $w'$ be the neighbor of $w$ distinct from $v$ and $x$, and $x'$ be the neighbor of $x$ distinct from $v$ and $w$. By Lemmas~\ref{notwotri} and~\ref{notrisqua}, $w'$, $x'$ and $u$ are distinct and pairwise not adjacent. Moreover, if $w' \notin V_1$ or $x' \notin V_1$, say $w' \notin V_1$, then $f = ww'$, and thus $xx'$ is a cut edge of $H$, a contradiction. Hence $v$, $w$, $x$, $w'$ and $x'$ are all in $V_1$, and thus, by $P_v$, $H-w$ is $2$-connected. Let $H' = H - \{v,w,x\} + ux'$. Graph $H'$ is in ${\cal C}_{2,3^-}$. By minimality of $H$, $H'$ admits a feedback vertex set $S'$ of size at most $\frac{n - 3 + 2}{3}$. The set $S = S' \cup \{w\}$ is a feedback vertex set of $H$ of size $|S'| + 1 \le \frac{n - 3 + 2}{3} + 1 \le \frac{n + 2}{3}$, a contradiction.

Let $w_0$ and $w_1$ be the two neighbors of $w$ distinct from $v$. If $w_0$ or $w_1$ is in $V_2$, say $w_0 \in V_2$, then $ww_0 = f$, and $\{vx,ww_1\}$ is an edge cut-set of $H$, contradicting the minimality of $|V_1|$. Therefore $w_0$ and $w_1$ are in $V_1$. 

Let us show that $w_0w_1 \notin E$. By contradiction assume that $w_0w_1 \in E$. Let $w_0'$ be the neighbor of $w_0$ distinct from $w$ and $w_1$, and $w_1'$ be the neighbor of $w_1$ distinct from $w$ and $w_0$. By Lemmas~\ref{notwotri} and~\ref{notrisqua}, $w_0'$ and $w_1'$ are distinct and not adjacent. Vertices $v$, $w$, $w_0$ and $w_1$ are all in $V_1$, thus, by $P_v$, $H-w$ is $2$-connected. Let $H' = H - \{w,w_0,w_1\} + w_0'w_1'$. Graph $H'$ is in ${\cal C}_{2,3^-}$. By minimality of $H$, $H'$ admits a feedback vertex set $S'$ of size at most $\frac{n - 3 + 2}{3}$. The set $S = S' \cup \{w\}$ is a feedback vertex set of $H$ of size $|S'| + 1 \le \frac{n - 3 + 2}{3} + 1 \le \frac{n + 2}{3}$, a contradiction.

Let $w_{00}$ and $w_{01}$ be the two neighbors of $w_0$ distinct from $w$. Let us show that $w_{00}w_{01} \notin E$. By contradiction assume that $w_{00}w_{01} \in E$. Let $w_{00}'$ be the neighbor of $w_{00}$ distinct from $w_0$ and $w_{01}$, and $w_{01}'$ be the neighbor of $w_{01}$ distinct from $w_0$ and $w_{00}$. By Lemmas~\ref{notwotri} and~\ref{notrisqua}, $w_{00}'$ and $w_{01}'$ are distinct and not adjacent. Suppose $w_{00}$ or $w_{01}$ is in $V_2$, say $w_{00} \in V_2$. Then $w_0w_{00} = f$, and ${e,f}$ is not an edge cut-set of $H$ (since $w_0w_{00}w_{01}$ is a triangle), a contradiction. Therefore $w$, $w_0$, $w_{00}$ and $w_{01}$ are in $V_1$, and thus, by $P_v$, $H-w_0$ is $2$-connected. Let $H' = H - \{w_0,w_{00},w_{01}\} + w_{00}'w_{01}'$. Graph $H'$ is in ${\cal C}_{2,3^-}$. By minimality of $H$, $H'$ admits a feedback vertex set $S'$ of size at most $\frac{n - 3 + 2}{3}$. The set $S = S' \cup \{w_0\}$ is a feedback vertex set of $H$ of size $|S'| + 1 \le \frac{n - 3 + 2}{3} + 1 \le \frac{n + 2}{3}$, a contradiction.

Let $w_{10}$ and $w_{11}$ be the two neighbors of $w_1$ distinct from $w$. By symmetry, $w_{10}w_{11} \notin E$. Suppose $\{w_{00}, w_{01}\} = \{w_{10}, w_{11}\}$; say $w_{00} = w_{10}$ and $w_{01} = w_{11}$. Lemma~\ref{notwosqua} leads to a contradiction. Therefore the pairs $\{w_{00}, w_{01}\}$ and $\{w_{10}, w_{11}\}$ are not equal. As $v$, $w$, $w_0$ and $w_1$ are in $V_1$, by $P_v$, $H-w$ is $2$-connected. Let $H' = H - \{w,w_0,w_1\} + \{w_{00}w_{01}, w_{10}w_{11}\}$. Graph $H'$ is in ${\cal C}_{2,3^-}$. By minimality of $H$, $H'$ admits a feedback vertex set $S'$ of size at most $\frac{n - 3 + 2}{3}$. The set $S = S' \cup \{w\}$ is a feedback vertex set of $H$ of size $|S'| + 1 \le \frac{n - 3 + 2}{3} + 1 \le \frac{n + 2}{3}$, a contradiction, which completes the proof.
\end{proof}

\begin{lemm}\label{notri}
There is no triangle in $H$.
\end{lemm}

\begin{proof}
Suppose there is a triangle $uvw$ in $H$. Let $u'$, $v'$ and $w'$ be the third neighbor of $u$, $v$ and $w$ respectively. By Lemmas~\ref{notwotri} and~\ref{notrisqua}, $u'$, $v'$ and $w'$ are distinct and non-adjacent. Let $H' = H - \{u,v,w\} + u'v'$. Observe that by Lemma~\ref{3co}, $H - w$ is $2$-connected. Therefore $H'$ is in ${\cal C}_{2,3^-}$. By minimality of $H$, $H'$ admits a feedback vertex set $S'$ of size at most $\frac{n - 3 + 2}{3}$. The set $S = S' \cup \{w\}$ is a feedback vertex set of $H$ of size $|S'| + 1 \le \frac{n - 3 + 2}{3} + 1 \le \frac{n + 2}{3}$, a contradiction.
\end{proof}

Let $v$ be a vertex of $H$, and $x$ and $y$ be two neighbors of $v$. They are not adjacent by Lemma~\ref{notri}. Let $x_0$, $x_1$, $y_0$ and $y_1$ the two other neighbors of $x$ and $y$ respectively. Vertices $x_0$ and $x_1$ are not adjacent by Lemma~\ref{notri}, and similarily $y_0$ and $y_1$ are not adjacent. The pairs $\{x_0, x_1\}$ and $\{y_0, y_1\}$ are distinct by Lemma~\ref{notwosqua}. Let $H' = H - \{v,x,y\} + \{x_0x_1, y_0y_1\}$. By Lemma~\ref{3co}, $H'$ is in ${\cal C}_{2,3^-}$. By minimality of $H$, $H'$ admits a feedback vertex set $S'$ of size at most $\frac{n - 3 + 2}{3}$. The set $S = S' \cup \{v\}$ is a feedback vertex set of $H$ of size $|S'| + 1 \le \frac{n - 3 + 2}{3} + 1 \le \frac{n + 2}{3}$, a contradiction. That completes the proof of Theorem~\ref{subth}.

\section{Proof of Theorem \ref{main} \label{proofmain}}

Let $g \ge 3$ be a fixed integer.
For $G$ a planar graph, $\omega: E(G) \rightarrow \mathbb{N}$ a weight function, and $F \subseteq E(G)$, we denote $\sum_{e \in F}(\omega(e))$ by $\omega(F)$, and $\sum_{e \in E(G)}(\omega(e))$ by $\omega(G)$. We will prove the following claim:

\begin{claim} \label{todo}
Let $G$ be a planar graph, and $\omega: E(G) \rightarrow \mathbb{N}$ a weight function such that for each cycle $C$ of $G$, $\omega(C) \ge g$. There exists a feedback vertex set $S$ of $G$ of size at most $\frac{4\omega(G)}{3g}$.
\end{claim}

Observe that fixing $\omega$ constant equal to $1$ in Claim \ref{todo} yields Theorem \ref{main}. Let us consider any embedding of the graph $G$ in the plane.
\medskip

Let $G$ be a $2$-connected plane graph. Three faces $f_0$, $f_1$ and $f_2$ of $G$ are said to be \emph{mergeable} if:
\begin{enumerate}
	\item there exists a vertex $v$ that is in the boundary of $f_0$, $f_1$ and $f_2$.
	\item w.l.o.g. $f_0$ and $f_1$ (resp $f_1$ and $f_2$) have at least one common edge in their boundary.
\end{enumerate}

Given three mergeable faces $f_0$, $f_1$ and $f_2$, the \emph{merger} of $f_0$, $f_1$ and $f_2$ consists in removing the edges belonging to the boundary of two faces among $f_0$, $f_1$ and $f_2$ as well as the vertices that end up being isolated. The common vertex $v$ of $f_0$, $f_1$ and $f_2$ is called the \emph{crucial} vertex of the merger. A merger is \emph{nice} if the sum of the weights of the edges removed is at least $\frac{3g}{4}$. Observe that a merger cannot decrease $\min_{C \text{ cycle of }G}(\omega(C))$, since we only delete vertices and edges. See Figure~\ref{figmerg} for an example of the merger of three faces.

\begin{figure}[h]
\begin{center}
\begin{tikzpicture}[line cap=round,line join=round,>=triangle 45,x=0.7cm,y=0.7cm]
\clip(-3.0,-5.0) rectangle (16.0,5.0);
\fill[color=ttqqqq,fill=ttqqqq,fill opacity=0.1] (-2.0,2.0) -- (0.0,2.0) -- (0.0,-0.0) -- (-2.0,0.0) -- cycle;
\fill[color=ttqqqq,fill=ttqqqq,fill opacity=0.1] (0.0,4.0) -- (3.6905989232414966,4.0) -- (6.0,0.0) -- (4.0,0.0) -- (3.0,1.7320508075688774) -- (2.0,0.0) -- (0.0,-0.0) -- cycle;
\fill[color=ttqqqq,fill=ttqqqq,fill opacity=0.1] (0.0,-0.0) -- (0.0,-2.0) -- (3.6905989232414966,-4.0) -- (6.0,0.0) -- (4.0,0.0) -- (3.0,-1.7320508075688772) -- (2.0,0.0) -- cycle;
\draw (0.0,-0.0)-- (-2.0,0.0);
\draw (0.0,-0.0)-- (0.0,2.0);
\draw (0.0,-0.0)-- (2.0,0.0);
\draw (0.0,-0.0)-- (0.0,-2.0);
\draw (4.0,0.0)-- (6.0,0.0);
\draw (0.0,4.0)-- (0.0,2.0);
\draw (0.0,2.0)-- (-2.0,2.0);
\draw (-2.0,2.0)-- (-2.0,0.0);
\draw (3.0,1.7320508075688774)-- (2.0,0.0);
\draw (3.0,1.7320508075688774)-- (4.0,0.0);
\draw (4.0,0.0)-- (3.0,-1.7320508075688772);
\draw (3.0,-1.7320508075688772)-- (2.0,0.0);
\draw (0.0,4.0)-- (3.6905989232414966,4.0);
\draw (3.6905989232414966,4.0)-- (6.0,0.0);
\draw (-2.0,0.0)-- (-2.0,-2.0);
\draw (-2.0,-2.0)-- (0.0,-2.0);
\draw (-2.0,2.0)-- (0.0,4.0);
\draw (0.0,-2.0)-- (3.6905989232414966,-4.0);
\draw (3.6905989232414966,-4.0)-- (6.0,0.0);
\draw (3.6905989232414966,4.0)-- (6.0,0.0);
\draw (-0.009474437408312709,0.56415870630746653) node[anchor=north west] {$v$};
\draw (-1.3153057460946183,1.387853093563085) node[anchor=north west] {$f_0$};
\draw (1.9029632376138863,2.623782203271307) node[anchor=north west] {$f_1$};
\draw (1.9810219182270372,-1.7344941309629496) node[anchor=north west] {$f_2$};
\begin{scriptsize}
\draw [fill=black] (0.0,-0.0) circle (1.5pt);
\draw [fill=black] (0.0,2.0) circle (1.5pt);
\draw [fill=black] (2.0,0.0) circle (1.5pt);
\draw [fill=black] (0.0,-2.0) circle (1.5pt);
\draw [fill=black] (-2.0,0.0) circle (1.5pt);
\draw [fill=black] (4.0,0.0) circle (1.5pt);
\draw [fill=black] (5.0,0.0) circle (1.5pt);
\draw [fill=black] (6.0,0.0) circle (1.5pt);
\draw [fill=black] (0.0,4.0) circle (1.5pt);
\draw [fill=black] (-2.0,2.0) circle (1.5pt);
\draw [fill=black] (3.0,1.7320508075688774) circle (1.5pt);
\draw [fill=black] (3.0,-1.7320508075688772) circle (1.5pt);
\draw [fill=black] (3.6905989232414966,4.0) circle (1.5pt);
\draw [fill=black] (-2.0,-2.0) circle (1.5pt);
\draw [fill=black] (3.6905989232414966,-4.0) circle (1.5pt);
\end{scriptsize}
\fill[color=ttqqqq,fill=ttqqqq,fill opacity=0.1] (7.0,2.0) -- (9.0,2.0) -- (9.0,4.0) -- (12.6905989232414966,4.0) -- (15.0,0.0) -- (13.0,0.0) -- (12.0,1.7320508075688774) -- (11.0,0.0) -- (12.0,-1.7320508075688772) -- (13.0,0.0) -- (15.0,0.0) -- (12.6905989232414966,-4.0) -- (9.0,-2.0) -- (9.0,0.0) -- (7.0,0.0) -- cycle;
\draw (9.0,-0.0)-- (7.0,0.0);
\draw (9.0,-0.0)-- (9.0,-2.0);
\draw (13.0,0.0)-- (13.0,0.0);
\draw (9.0,4.0)-- (9.0,2.0);
\draw (9.0,2.0)-- (7.0,2.0);
\draw (7.0,2.0)-- (7.0,0.0);
\draw (12.0,1.7320508075688774)-- (11.0,0.0);
\draw (12.0,1.7320508075688774)-- (13.0,0.0);
\draw (13.0,0.0)-- (12.0,-1.7320508075688772);
\draw (12.0,-1.7320508075688772)-- (11.0,0.0);
\draw (9.0,4.0)-- (12.6905989232414966,4.0);
\draw (12.6905989232414966,4.0)-- (15.0,0.0);
\draw (7.0,0.0)-- (7.0,-2.0);
\draw (7.0,-2.0)-- (9.0,-2.0);
\draw (7.0,2.0)-- (9.0,4.0);
\draw (9.0,-2.0)-- (12.6905989232414966,-4.0);
\draw (12.6905989232414966,-4.0)-- (15.0,0.0);
\draw (12.6905989232414966,4.0)-- (15.0,0.0);
\draw (9,0.56415870630746653) node[anchor=north west] {$v$};
\draw (9.927229729949499,0.2039631042636303) node[anchor=north west] {$f$};
\begin{scriptsize}
\draw [fill=black] (9.0,-0.0) circle (1.5pt);
\draw [fill=black] (9.0,2.0) circle (1.5pt);
\draw [fill=black] (11.0,0.0) circle (1.5pt);
\draw [fill=black] (9.0,-2.0) circle (1.5pt);
\draw [fill=black] (7.0,0.0) circle (1.5pt);
\draw [fill=black] (13.0,0.0) circle (1.5pt);
\draw [fill=black] (15.0,0.0) circle (1.5pt);
\draw [fill=black] (9.0,4.0) circle (1.5pt);
\draw [fill=black] (7.0,2.0) circle (1.5pt);
\draw [fill=black] (12.0,1.7320508075688774) circle (1.5pt);
\draw [fill=black] (12.0,-1.7320508075688772) circle (1.5pt);
\draw [fill=black] (12.6905989232414966,4.0) circle (1.5pt);
\draw [fill=black] (7.0,-2.0) circle (1.5pt);
\draw [fill=black] (12.6905989232414966,-4.0) circle (1.5pt);
\end{scriptsize}
\end{tikzpicture}

\end{center}
\caption{The merger of faces $f_0$, $f_1$ and $f_2$ into $f$ with crucial vertex $v$. \label{figmerg}}
\end{figure}
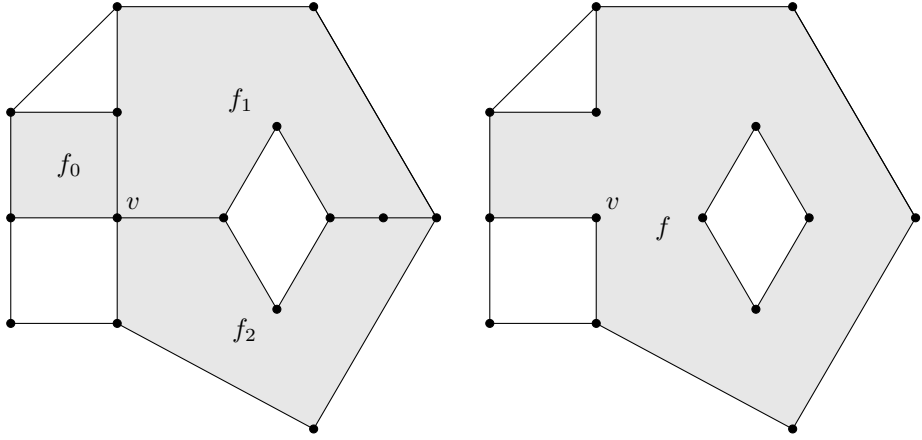

\begin{lemm} \label{merge}
Let $G$ be a $2$-connected plane graph, and $G'$ obtained from $G$ by applying a merger of crucial vertex $v$. If $S'$ is a feedback vertex set of $G'$, then $S' \cup \{v\}$ is a feedback vertex set of $G$.
\end{lemm}

\begin{proof} 
Let $C$ be a cycle of $G$ that contains an edge $e \in E(G) \backslash E(G')$. Edge $e$ is in the boundary of two of the faces that are merged, say $f_0$ and $f_1$. Cycle $C$ separates $f_0$ and $f_1$. Therefore it contains all the vertices of $V(G[f_0]) \cap V(G[f_1])$. In particular, it contains $v$.

Therefore each cycle of $G$ is either entirely in $G'$, or it contains $v$. Thus as $V(G') \backslash S'$ induces a forest in $G'$, $V(G) \backslash (S' \cup \{v\})$ induces a forest in $G$.
\end{proof}

\begin{lemm} \label{nice}
Let $G$ be a $2$-connected plane graph, and $G'$ be obtained from $G$ by applying a nice merger. If graph $G'$ satisfies Claim \ref{todo}, then graph $G$ also satisfies Claim \ref{todo}. 
\end{lemm}

\begin{proof}
Let $v$ be the crucial vertex of the merger. We have $\omega(G') \le \omega(G) - \frac{3g}{4}$. Since $G'$ verifies Claim \ref{todo}, there exists a feedback vertex set $S'$ of $G'$ such that $|S'| \le \frac{4\omega(G')}{3g} \le \frac{4\omega(G)}{3g} - 1$. Then $S = S' \cup \{v\}$ is a feedback vertex set of $G$ (by Lemma \ref{merge}), and $|S| \le \frac{4\omega(G)}{3g} - 1 + 1 = \frac{4\omega(G)}{3g}$, which completes the proof.
\end{proof}

Let us assume by contradiction that there are couples $(G,\omega)$ that do not satisfy Claim~\ref{todo}. Among all counterexamples $(G,\omega)$ to Claim~\ref{todo} minimizing $\omega(G)$, we consider a couple $(G,\omega)$ minimizing $\sum_{v \in V(G)} (\max\{0.5,d(v)-2.5\})$.

\begin{lemm} \label{2co}
Graph $G$ is $2$-connected.
\end{lemm}

\begin{proof}
By contradiction, assume $G$ is not $2$-connected. Graph $G$ has at least $2$ vertices, otherwise it would satisfy Claim~\ref{todo}.
Let $S$ be a minimal vertex cut-set of $G$. We have $|S| \le 1$. Let $V_1$ and $V_2$ be non-empty sets of vertices separated by $S$.

Let $\omega_1 = \omega(G[V_1 \cup S])$ and $\omega_2 = \omega(G[V_2 \cup S])$. By minimality of $(G,\omega)$, let $S_1 \subseteq V_1 \cup S$ and $S_2 \subseteq V_2 \cup S$ be feedback vertex sets of $V_1 \cup S$ and $V_2 \cup S$ respectively, such that $|S_1| \le \frac{4\omega_1}{3g}$ and $|S_2| \le \frac{4\omega_2}{3g}$. Now $S_1 \cup S_2$ is a feedback vertex set of $G$, and $|S_1 \cup S_2| \le \frac{4\omega_1}{3g} + \frac{4\omega_2}{3g} = \frac{4\omega(G)}{3g}$. Thus $G$ satisfies Claim \ref{todo}, a contradiction.
\end{proof}

\begin{lemm} \label{nice2}
No nice mergers can be done in $G$.
\end{lemm}

\begin{proof}
It follows from Lemma \ref{nice} and the minimality of $(G,\omega)$.
\end{proof}

\begin{lemm} \label{nice3}
Every face in $G$ has at least three $3^+$-vertices in its boundary.
\end{lemm}

\begin{proof}
Let us assume that there is a face $f$ in $G$ with at most two $3^+$-vertices in its boundary. Face $f$ is adjacent to at most two other faces in $G$. Suppose $f$ is adjacent to exactly one face, say $f'$. As $G$ is $2$-connected by Lemma~\ref{2co}, $G[f]$ and $G[f']$ are cycles. As $f$ is adjacent only to $f'$, $E(G[f]) \subseteq E(G[f'])$, and thus $G[f] = G[f']$. So two faces of $G$ have exactly the same boundary, so $G$ is a cycle, and it satisfies Claim~\ref{todo}, a contradiction.

Thus $f$ is adjacent to exactly two other faces, say $f_0$ and $f_1$. Then $E(G[f]) \subseteq E(G[f_0]) \cup E(G[f_1])$, and  $E(G[f]) \cap E(G[f_0]) \ne \emptyset \ne E(G[f]) \cap E(G[f_1])$. As $G[f]$ is a cycle, there is a vertex $v$ in $V(G[f])$ incident to an edge in $E(G[f]) \cap E(G[f_0])$ and to an edge in $E(G[f]) \cap E(G[f_1])$. Merging the faces $f$, $f_0$ and $f_1$ with crucial vertex $v$ is nice, since we remove all the edges of $G[f]$ and $\omega(f) \ge g \ge \frac{3g}{4}$. This leads to a contradiction with Lemma~\ref{nice2}.
\end{proof}

\begin{lemm} \label{degle3}
There are no $4^+$-vertices in $G$.
\end{lemm}

\begin{proof}
Suppose $v$ is a $d$-vertex in $G$ with $d \ge 4$. Let $u_0$, ..., $u_{d-1}$ be the neighbors of $v$. Let $G' = G - v + \{w,w'\} + \{wu_0, wu_1, ww', w'u_2, ..., w'u_{d-1}\}$, $\omega(wu_0) = \omega(vu_0)$, $\omega(wu_1) = \omega(vu_1)$, $\omega(w'u_2) = \omega(vu_2)$, ... , $\omega(w'u_{d-1}) = \omega(vu_{d-1})$, and $\omega(ww') = 0$. See Figure~\ref{figdegle3} for an illustration of this construction. Clearly, $\omega(G') = \omega(G)$. As we removed a $d$-vertex, added a $3$-vertex and a $(d-1)$-vertex, and did not change the degree of the other vertices, $\sum_{v \in V(G')} (\max\{0.5,d(v)-2.5\}) = \sum_{v \in V(G)} (\max\{0.5,d(v)-2.5\}) - 0.5$.


It is easy to see that for any cycle $C'$ of $G'$, there is a cycle in $G$ that has the same weight, so $\omega(C') \ge g$.

By minimality of $(G,\omega)$, let $S'$ be a feedback vertex set of $G'$ with $|S'| \le \frac{4\omega(G')}{3}$. For any cycle $C$ of $G$ there is a cycle $C'$ of $G'$ such that $C = C'$ or $V(C) = (V(C') \backslash \{w,w'\}) \cup \{v\}$. If $w \in S'$ or $w' \in S'$, then let $S = S' \backslash \{w,w'\} \cup \{v\}$ and otherwise let $S = S'$. Then $|S| \le |S'| \le \frac{4\omega(G')}{3} = \frac{4\omega(G)}{3}$, and $S$ is a feedback vertex set of $G$, a contradiction.
\end{proof}

\begin{figure}[h]
\begin{center}
\begin{tikzpicture}[line cap=round,line join=round,>=triangle 45,x=1.0cm,y=1.0cm]
\clip(-3.0,-3.0) rectangle (9.5,3.0);
\draw (0.0,-0.0)-- (0.10937499999999957,1.9970070378881994);
\draw (0.0,-0.0)-- (1.0625,1.694430213965745);
\draw (0.0,-0.0)-- (1.75,0.9682458365518543);
\draw (0.0,-0.0)-- (2.0,0.0);
\draw (0.0,-0.0)-- (0.0,-2.0);
\draw (0.0,-0.0)-- (-2.0,0.0);
\draw (-0.3882014954742202,0.6470759543486791) node[anchor=north west] {$v$};
\draw (-2.2142542306178643,0.615592286501375) node[anchor=north west] {$u_0$};
\draw (-0.63542699724517653,-1.3993624557260893) node[anchor=north west] {$u_1$};
\draw (1.6582369146005531,0.048886265249900696) node[anchor=north west] {$u_2$};
\draw (1.5795277449822926,0.9776544667453725) node[anchor=north west] {$u_3$};
\draw (0.9970798898071649,1.7490043290043238) node[anchor=north west] {$u_4$};
\draw (0.1,2.079582841401017) node[anchor=north west] {$u_{d-1}$};
\draw (0.1627626918536034,1.3) node[anchor=north west] {$...$};
\draw (5.701828216053249,-0.6858392365210557)-- (7.154196970881031,0.6891607634789452);
\draw (7.154196970881031,0.6891607634789452)-- (7.49051259346714,1.6960909774446895);
\draw (7.154196970881031,0.6891607634789452)-- (6.53738759346714,1.998667801367144);
\draw (7.154196970881031,0.6891607634789452)-- (8.178012593467141,0.9699066000307988);
\draw (7.154196970881031,0.6891607634789452)-- (8.428012593467141,0.0016607634789444982);
\draw (5.701828216053249,-0.6858392365210557)-- (4.42801259346714,0.0016607634789444982);
\draw (5.701828216053249,-0.6858392365210557)-- (6.42801259346714,-1.9983392365210555);
\draw (4.192672176308542,0.615592286501375) node[anchor=north west] {$u_0$};
\draw (5.755757575757577,-1.6354899645808703) node[anchor=north west] {$u_1$};
\draw (5.094600550964189,-0.5335615899252257) node[anchor=north west] {$w$};
\draw (6.810460448642268,0.6887524596615477) node[anchor=north west] {$w'$};
\draw (6.54839826839827,2.33145218417945) node[anchor=north west] {$u_{d-1}$};
\draw (6.8,1.7) node[anchor=north west] {$...$};
\draw (7.5771664698937435,1.9064226682408443) node[anchor=north west] {$u_4$};
\draw (8.175356158992523,1.0406218024399807) node[anchor=north west] {$u_3$};
\draw (8.191097992916175,0.017402597402596563) node[anchor=north west] {$u_2$};
\begin{scriptsize}
\draw [fill=black] (0.0,-0.0) circle (1.5pt);
\draw [fill=black] (2.0,0.0) circle (1.5pt);
\draw [fill=black] (0.0,-2.0) circle (1.5pt);
\draw [fill=black] (-2.0,0.0) circle (1.5pt);
\draw [fill=black] (1.75,0.9682458365518543) circle (1.5pt);
\draw [fill=black] (1.0625,1.694430213965745) circle (1.5pt);
\draw [fill=black] (0.10937499999999957,1.9970070378881994) circle (1.5pt);
\draw [fill=black] (6.53738759346714,1.998667801367144) circle (1.5pt);
\draw [fill=black] (7.49051259346714,1.6960909774446895) circle (1.5pt);
\draw [fill=black] (8.178012593467141,0.9699066000307988) circle (1.5pt);
\draw [fill=black] (8.428012593467141,0.0016607634789444982) circle (1.5pt);
\draw [fill=black] (6.42801259346714,-1.9983392365210555) circle (1.5pt);
\draw [fill=black] (4.42801259346714,0.0016607634789444982) circle (1.5pt);
\draw [fill=black] (7.154196970881031,0.6891607634789452) circle (1.5pt);
\draw [fill=black] (5.701828216053249,-0.6858392365210557) circle (1.5pt);
\end{scriptsize}
\end{tikzpicture}
\end{center}
\caption{The construction of Lemma \ref{degle3}. \label{figdegle3}}
\end{figure}
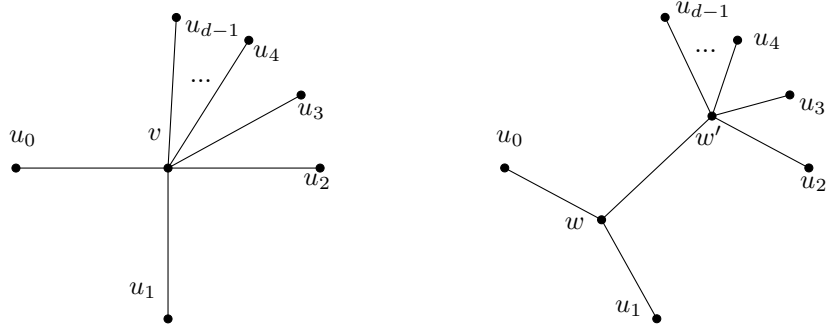

\begin{lemm} \label{nice4}
Every cycle has at least three $3$-vertices in $G$.
\end{lemm}

\begin{proof}
Let $C$ be a cycle of $G$. By Lemma \ref{degle3}, every vertex in $V(C)$ has degree at most $3$. Suppose $C$ is a separating cycle. By Lemma \ref{2co}, graph $G$ is $2$-connected, so at least two vertices of $V(C)$ have a neighbor in the interior of $C$, and at least two vertices of $V(C)$ have a neighbor in the exterior of $C$. Therefore $C$ has at least four $3$-vertices. Now if $C$ bounds a face, then Lemma \ref{nice3} concludes the proof.
\end{proof}

\begin{figure}[h]
\begin{center}
\begin{tikzpicture}[line cap=round,line join=round,>=triangle 45,x=0.8cm,y=0.8cm]
\clip(-4.5,-4.0) rectangle (14.0,2.0);
\draw (2.0,0.0)-- (0.0,-0.0);
\draw (0.0,-0.0)-- (-1.6267960970477846,-1.1634149984550204);
\draw (2.0,0.0)-- (3.517540262479736,1.302716988356617);
\draw (2.0,0.0)-- (3.517540262479736,-1.302716988356617);
\draw (-1.6267960970477846,-1.1634149984550204)-- (-2.658276307936857,-2.876905171915094);
\draw (-1.6267960970477846,-1.1634149984550204)-- (-3.626562581098179,-1.1328534332345988);
\draw (10.661299167169831,0.010806375986089925)-- (12.178839429649567,1.313523364342707);
\draw (10.661299167169831,0.010806375986089925)-- (12.178839429649567,-1.291910612370527);
\draw (7.034503070122047,-1.1526086224689305)-- (6.003022859232974,-2.866098795929004);
\draw (7.034503070122047,-1.1526086224689305)-- (5.034736586071652,-1.1220470572485088);
\draw (7.034503070122047,-1.1526086224689305)-- (10.661299167169831,0.010806375986089925);
\draw (-0.10846359730812956,0.624076499376157) node[anchor=north west] {$v$};
\draw (-1.8102881897155663,-0.4798097227259637) node[anchor=north west] {$u$};
\draw (0.25,0.624076499376157) node[anchor=north west] {$\omega(vw)$};
\draw (-1.0,-0.4) node[anchor=north west] {$\omega(uv)$};
\draw (1.7006832666925689,0.624076499376157) node[anchor=north west] {$w$};
\draw (6.867484056253885,-0.44914621655646036) node[anchor=north west] {$u$};
\draw (7.9,-0.75) node[anchor=north west] {$\omega(uv) + \omega(vw)$};
\draw (10.347792006492517,0.624076499376157) node[anchor=north west] {$w$};
\begin{scriptsize}
\draw [fill=ttqqqq] (0.0,-0.0) circle (1.5pt);
\draw [fill=ttqqqq] (2.0,0.0) circle (1.5pt);
\draw [fill=ttqqqq] (-1.6267960970477846,-1.1634149984550204) circle (1.5pt);
\draw [fill=ttqqqq] (10.661299167169831,0.010806375986089925) circle (1.5pt);
\draw [fill=ttqqqq] (7.034503070122047,-1.1526086224689305) circle (1.5pt);
\end{scriptsize}
\end{tikzpicture}
\end{center}
\caption{The construction of Lemma \ref{3reg}. \label{fig3reg}}
\end{figure}
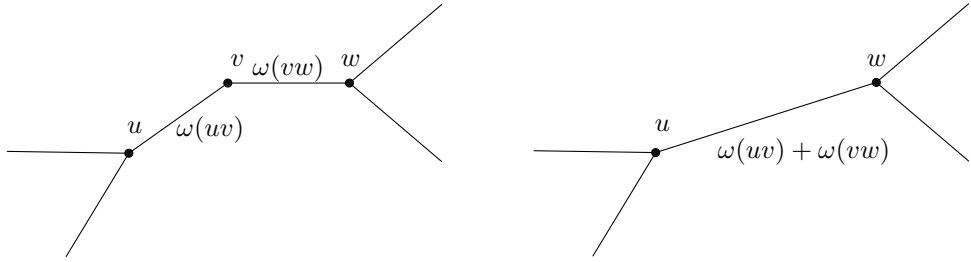

\begin{lemm} \label{3reg}
Graph $G$ is cubic (i.e. $3$-regular).
\end{lemm}

\begin{proof}
Suppose $v$ is a $2^-$-vertex in $G$. Vertex $v$ has degree $2$ by Lemma \ref{2co}. Let $u$ and $w$ be the two neighbors of $v$. By lemma \ref{nice4}, $uw \notin E(G)$.

Let $G' = G - v + uw$ and $\omega(uw) = \omega(uv) + \omega(vw)$. See Figure~\ref{fig3reg} for an illustration of this construction. Clearly, $\omega(G') = \omega(G)$. As we removed a $2$-vertex and did not change the degree of the other vertices, $\sum_{v \in V(G')} (\max\{0.5,d(v)-2.5\}) = \sum_{v \in V(G)} (\max\{0.5,d(v)-2.5\}) - 0.5$. 

Let $C'$ be any cycle of $G'$. If $uw \notin E(C')$, then $C'$ is a cycle of $G$, and so $\omega(C') \ge g$. Otherwise, $C = C' - uw + v + \{uv,vw\}$ is a cycle of $G$, and $\omega(C)= \omega(C')$, so $\omega(C') \ge g$.

For any cycle $C$ of $G$ there is a cycle $C'$ of $G'$ that contains all the vertices of $V(C) \backslash \{v\}$. By minimality of $(G,\omega)$, let $S'$ be a feedback vertex set of $G'$ with $|S'| \le \frac{4\omega(G')}{3} = \frac{4\omega(G)}{3}$. The set $S'$ is a feedback vertex set of $G$, a contradiction.
\end{proof}

By Lemmas~\ref{2co} and~\ref{3reg}, graph $G$ is a $2$-connected cubic graph. By Theorem~\ref{subth}, $G$ admits a feedback vertex set of order at most $\frac{|V(G)|+2}{3}$. Let us denote by $n$ the order of $G$, by $m$ the size of $G$ and by $f$ the number of faces of $G$.

By Euler's formula, we have $n - m + f = 2$. We have $3n = 2m$ as $G$ is cubic. Therefore, $f = 2+m-n = 2 + \frac{n}{2}$, i.e. $n = 2(f-2)$. Therefore $G$ has a feedback vertex set $S$ of size $|S| \le \frac{2f-4 + 2}{3} \le \frac{2f}{3}$. As each face has weight at least $g$, we have $gf \le 2 \omega(G)$, so $|S| \le \displaystyle{\frac{4\omega(G)}{3g}}$, a contradiction, completing the proof of Theorem~\ref{main}.

\bibliographystyle{plain}
\bibliography{biblio} {}

\end{document}